\documentclass[aps,pra,11pt,onecolumn,nofootinbib,tightenlines]{revtex4-2}

\usepackage{mathrsfs}
\usepackage{physics}
\usepackage{graphicx}
\usepackage{comment}
\usepackage{amsmath}
\usepackage{amsfonts}
\usepackage{amssymb}
\usepackage{braket}
\usepackage{amsthm}
\usepackage{caption}
\usepackage{lipsum}
\usepackage{subcaption}
\usepackage{mwe}
\usepackage{enumerate}
\usepackage{appendix}
\usepackage{tikz}
\usepackage{enumitem}
\usepackage{booktabs}
\usepackage{array}
\usepackage{forest}
\usepackage{mathtools}
\usepackage[colorlinks=true,urlcolor=blue,citecolor=blue,linkcolor=blue]{hyperref}

\usepackage[most,breakable]{tcolorbox}

\newtheorem{theorem}{Theorem}[section]
\newtheorem{lemma}[theorem]{Lemma}
\newtheorem{proposition}[theorem]{Proposition}

\newtheorem{remark}[theorem]{Remark}
\let\oldremark\remark
\renewcommand{\remark}{\oldremark\upshape}

\newtheorem{corollary}[theorem]{Corollary}

\DeclareMathOperator*{\opt}{opt}
\DeclareMathOperator*{\Trm}{Tr}

\DeclareMathOperator*{\id}{id}

\begin{document}

\title{Sharp Data-Processing Region for Quantum Conditional R\'enyi Entropies}
\author{Milad M. Goodarzi$^1$}
\email{milad.moazami@gmail.com}
\author{Yucheng Bao$^3$}
\author{Marco Tomamichel$^{1,2}$}

\affiliation{
$^1$Centre for Quantum Technologies, National University of Singapore, Singapore 117543, Singapore\\
$^2$Department of Electrical and Computer Engineering, National University of Singapore, Singapore 117583, Singapore\\
$^3$School of Mathematics and Physics, Xi’an Jiaotong-Liverpool University, Suzhou 215123, China
}

\begin{abstract}
{\bf Abstract.} We study the three-parameter family of quantum conditional R\'enyi entropies $H_{\alpha,z}^\lambda$ introduced in \textit{Commun. Math. Phys.} 407, 180 (2026). We prove that the previously established data-processing region is sharp. We then characterize the club-sandwiched family through a classical coarse-graining principle: within the full data-processing region, universal monotonicity under deterministic coarse-graining of the classical conditioned register holds if and only if $z = \alpha$. Finally, we establish rigidity of entropic duality by showing that any universal duality relation for finite-dimensional tripartite pure states must satisfy the known parameter relations. Thus, the data-processing region, the club-sandwiched subfamily, and the duality transformation are each uniquely characterized by their respective universal properties.
\end{abstract}

\maketitle

\section{Introduction}

Conditional R\'enyi entropies quantify uncertainty about a quantum system in the presence of correlated side information. They are particularly useful in non-asymptotic information theory and in the analysis of error and strong converse exponents. A three-parameter family $H_{\alpha,z}^\lambda$ was recently introduced in \cite{rubboli2026quantum} which unifies a large number of previously studied classical and quantum conditional entropies. This family has already appeared in operational applications to privacy amplification, soft covering, and composable randomness extraction \cite{li2025two,rubboli2026strong}.

We work throughout with
\begin{equation}\label{parameter space}
    \alpha > 0, \qquad \alpha \neq 1, \qquad z > 0, \qquad \lambda \in \mathbb{R}.
\end{equation}

For a bipartite state $\rho_{AB}$, with marginal $\rho_B$, the quantity $H_{\alpha,z}^\lambda(A|B)_\rho$ is defined by
\begin{equation}\label{definition}
    H_{\alpha,z}^{\lambda}(A|B)_\rho = \frac{1}{1-\alpha} \opt_{\sigma_B} \log \Tr\left(\rho_{AB}^{\frac{\alpha}{2z}}\left(I_A\otimes\rho_B^{\frac{(1-\lambda)(1-\alpha)}{2z}}\sigma_B^{\frac{\lambda(1-\alpha)}{z}}\rho_B^{\frac{(1-\lambda)(1 - \alpha)}{2z}}\right)\rho_{AB}^{\frac{\alpha}{2z}}\right)^z,
\end{equation}
where $\opt$ is an infimum (resp.\ supremum) when $\lambda(1-\alpha)$ is negative (resp.\ non-negative). The family satisfies data processing, additivity, duality, chain rules, and parameter monotonicity in suitable parameter ranges \cite{rubboli2026quantum}. It has also found direct applications in information-theoretic error-exponent problems. In \cite{li2025two}, the authors used its classical specialization to characterize strong-converse exponents for privacy amplification and soft covering under R\'enyi divergence criteria. More recently, it was shown that they determine the exact strong converse exponent for composable randomness extraction against quantum side information \cite{rubboli2026strong} as well as the strong converse exponent for the entanglement cost of quantum state merging \cite{berta2026strong}. These developments make it important to determine whether the structural properties and parameter regions established in \cite{rubboli2026quantum} are optimal.

The present work addresses three such questions concerning data processing, monotonicity under coarse-graining, and duality. The data-processing theorem in \cite{rubboli2026quantum} is stronger than monotonicity under a channel on $B$ alone. An admissible conditioned operation has the form
\begin{equation}\label{conditional operation}
    \mathcal{E}_{AB \rightarrow A'B'}(\cdot) = \sum_i \mathcal{M}_{A \rightarrow A'}^i \otimes \mathcal{N}_{B \rightarrow B'}^i(\cdot),
\end{equation}
where every $\mathcal{M}^i$ is a sub-unital channel and $\{\mathcal{N}^i\}_i$ is a quantum instrument. We say that $H_{\alpha,z}^\lambda$ satisfies the \emph{data-processing inequality} (DPI) if
\begin{equation}\label{data processing}
    H_{\alpha,z}^\lambda(A|B)_\rho \leq H_{\alpha,z}^\lambda(A'|B')_{\mathcal{E}(\rho)},
\end{equation}
for every state and every operation \eqref{conditional operation}. A channel on the side information is the special case with one summand and $\mathcal{M} = \id_A$. Moreover, the class of transformations in \eqref{conditional operation} contains one-way LOCC protocols form $B$ to $A$ in which the outcome of an instrument on $B$ determines a conditional mixing channel (for example, a unitary channel) on $A$.

The region in which \eqref{data processing} was established in \cite[Theorem 5.2]{rubboli2026quantum} is $\mathcal{D} = \mathcal{D}_1 \cup \mathcal{D}_2$, where
\begin{align}
\mathcal{D}_1 &:= \left\{ (\alpha, z, \lambda) \in \mathbb{R}^3 : 0 < \alpha < 1, \quad 1-z \leq \alpha \leq z, \quad 1-\frac{z}{1-\alpha} \leq \lambda \leq 1 \right\},\label{cconstraints1}\\
\mathcal{D}_2 &:= \left\{ (\alpha, z, \lambda) \in \mathbb{R}^3 : 1 < \alpha < \infty, \quad \alpha-1 \leq z \leq \alpha \leq 2z, \quad 1+\frac{z}{1-\alpha} \leq \lambda \leq 1 \right\}.\label{constraints2}
\end{align}

Our first main result (Theorem \ref{Sharp DPI region}) shows that this region is sharp, as conjectured in \cite[Section 13]{rubboli2026quantum}. More precisely, universal validity of \eqref{data processing} forces $(\alpha,z,\lambda) \in \mathcal{D}$; combined with the sufficiency result \cite[Theorem 5.2]{rubboli2026quantum}, this gives an exact characterization of the DPI region. Thus, none of the constraints defining $\mathcal{D}$ is an artifact of the proof of sufficiency.

Determining the exact data-processing region of parametrized quantum information measures has an important precedent in the theory of quantum R\'enyi divergences. For the $\alpha$-$z$ R\'enyi divergences introduced by Audenaert and Datta \cite{audenaert13_alphaz}, substantial parts of the region were obtained from convexity and concavity properties of the underlying trace functionals, building in particular on work of Hiai \cite{hiai2013concavity,hiai2016concavity}. Further progress was made by Carlen, Frank, and Lieb \cite{carlen2016some,carlen18_alphaz}, before Zhang \cite{zhang20_alphaz} established the exact $(\alpha,z)$ data-processing region. As we explain next, this characterization is a key ingredient in determining the corresponding region for the three-parameter conditional entropies.

A central ingredient is a transfer principle (Theorem \ref{transfer}) that may be of independent interest: data processing for the conditional entrpy forces data processing for the corresponding $\alpha$-$z$ divergence. Recall \cite{audenaert13_alphaz} that
\begin{equation}
    D_{\alpha,z}(\rho\|\sigma) = \frac{1}{\alpha-1}\log{\Trm\big(\rho^\frac{\alpha}{2z}\sigma^\frac{1-\alpha}{z}\rho^\frac{\alpha}{2z}\big)^z}.
\end{equation}
We prove that if $H_{\alpha,z}^\lambda$ satisfies \eqref{data processing} for every state and every channel on the conditioning system, then
\begin{equation}
    D_{\alpha,z}\big(\mathcal{N}(\rho)\|\mathcal{N}(\sigma)\big) \leq D_{\alpha,z}(\rho\|\sigma)
\end{equation}
for every pair of states and every quantum channel $\mathcal{N}$. The conclusion is striking because the conditional entropy depends on the additional parameter $\lambda$ and explicitly involves the state's own marginal $\rho_B$, whereas the divergence contains neither. The proof uses a small-parameter family of classical--quantum states whose marginal on the conditioning system is fixed. The first nontrivial coefficient in the optmized conditional trace functional recovers the $\alpha$-$z$ trace functional; a decomposition argument then isolates an arbitrary pair $\rho,\sigma$. Zhang's exact characterization of data processing for $D_{\alpha,z}$ \cite{zhang20_alphaz} consequently yields the sharp restrictions on $z$. The remaining bounds on $\lambda$ follow from classical restrictions of \eqref{definition} together with monotonicity under pure-state transformations.

Our second main result (Theorem \ref{coarse graining thm}) concerns a complementary form of data processing, in which the classical conditioned system rather than the conditioning system is processed. This result gives a structural characterization of the club-sandwiched family through a natural compatibility requirement with classical post-processing. Consider a classical--quantum state $\rho_{XB}$ and a deterministic function $f$ acting of the classical register $X$, with $Z = f(X)$. Such a map merges classical outcomes and therefore discards informaton contained in $X$. We say that a conditional entropy satisfies universal coarse-graining if this loss of classical information can only decrease the conditional entropy,
\begin{equation}
    H_{\alpha,z}^\lambda(X|B)_\rho \geq H_{\alpha,z}^\lambda(Z|B)_{(R_f\otimes\operatorname{id}_B)(\rho)},
\end{equation}
for every finite classical--quantum state $\rho_{XB}$ and every deterministic function $f$, where $R_f$ is the classical channel induced by $f$.

We show that, within the full data-processing region $\mathcal{D}$, this property holds universally if and only if $z = \alpha$. Thus among all conditional entropies $H_{\alpha,z}^\lambda$ satisfying data processing, the club-sandwiched surface is uniquely distinguished by monotonicity under deterministic coarse-graining of the classical system. In this sense, the condition $z = \alpha$ is not merely a convenient algebraic specialization, but exactly the one enforced by universal compatibility with coarse-garining. Furthermore, combined with the duality theorem, this characterization has a particularly sharp consequence. Suppose that $H_{\alpha,z}^\lambda$ and its dual $H_{\hat{\alpha},\hat{z}}^{\hat{\lambda}}$ both satisfy the coarse-graining inequality. The present result then yields $z = \alpha$ and $\hat{\alpha} = \hat{z}$, while the duality relations \eqref{parameter relation} immediately imply $\lambda = \hat{\lambda} = 1$ and $1/\alpha + 1/\hat{\alpha} =2$. Thus, the optimized sandwiched conditional R\'enyi etropies $\widetilde{H}_\alpha^\uparrow$ are precisely the members of the three-parameter family for which coarse-graining holds simultaneously for the entropy and its dual. These entropies were introduced in \cite{lennert13_renyi}, and their structural properties were subsequently developed in \cite{beigi13_sandwiched,tomamichel13_duality,dupuis2015chain,beigi2023operator}. They have found important operational applications in strong-converse analyses of quantum information-processing tasks \cite{leditzky2016strong}, entropy accumulation and quantum cryptography \cite{dupuis20_entropyacc}, and, more recently, the exact characterization of the strong-converse exponent for quantum privacy amplification \cite{li2024operational}; see also \cite{berta2026tight} for a recent application.

Finally, the third main result of the paper (Theorem \ref{duality-converse}) establishes rigidity of entropic duality. The duality theorem of \cite{rubboli2026quantum} states that, for parameter triples $(\alpha,z,\lambda),(\hat{\alpha},\hat{z},\hat{\lambda}) \in \mathcal{D}$ satisfying
\begin{equation}\label{parameter relation}
    \frac{z}{1-\alpha}+\frac{\hat{z}}{1-\hat{\alpha}}=0, \qquad \frac{1-z}{1-\alpha}=\hat{\lambda}, \qquad \frac{1-\hat{z}}{1-\hat{\alpha}}=\lambda,
\end{equation}
one has
\begin{equation}\label{duality relation}
H_{\alpha,z}^{\lambda}(A|B)_\rho + H_{\hat{\alpha},\hat{z}}^{\hat{\lambda}}(A|C)_\rho = 0
\end{equation}
for every pure state $\rho_{ABC}$. We prove the converse: if \eqref{duality relation} holds for every finite-dimensional tripartite pure state, then all three relations in \eqref{parameter relation} are necessary. Hence the duality map of \cite{rubboli2026quantum} is rigid; there are no additional universal duality pairings hidden inside the three-parameter family.

The rigidity proof is driven by two complementary classes of test states. A bipartite pure state tensored with a trivial third system reduces the identity to equality of ordinary R\'enyi entropies and determines two combinations of the parameters. A second family of pure states, chosen so that one reduced state has a classical block structure while the complementary reduction contains an entangled block, supplies the remaining independent relation.


\section{Preliminaries}

Throughout, all Hilbert spaces are finite-dimensional. We write $\mathcal{S}(A)$ for the quantum states on $A$, and $\mathcal{S}(A)_{++}$ for faithful states.

For a state $\rho$ and $\alpha \in (0,\infty) \backslash \{1\}$, its R\'enyi entropy is
\begin{equation}
H_{\alpha}(\rho) := \frac{1}{1-\alpha} \log\Tr(\rho^{\alpha}).
\end{equation}

We use the continuous extensions $H_0(\rho) = \log\rank\rho$, $H_1(\rho) = -\Tr \rho\log\rho$, and $H_\infty(\rho) = -\log\|\rho\|_\infty$. For a faithful state the same formula defines $H_\alpha$ for every real $\alpha$.

The trace functionals used below are
\begin{align}
    \Phi_{p,q,r,s}(A,B,C) &:= \Tr\left(A^{p/2} B^{q/2} C^r B^{q/2} A^{p/2}\right)^s,\\
    \Psi_{a,b}(A\|B) &:= \Tr\left(A^{a/2b} B^{(1 - a)/b} A^{a/2b}\right)^b.
\end{align}
Thus $D_{a,b}(A\|B) = 1/(a - 1) \log\Psi_{a,b}(A\|B)$ for normalized states. With
\begin{equation}\label{parameters2}
    p = \frac{\alpha}{z}, \qquad q = \frac{(1 - \lambda)(1 - \alpha)}{z}, \qquad r = \frac{\lambda(1 - \alpha)}{z}, \qquad s = z,
\end{equation}
the trace functional in \eqref{definition} is $\Phi$ after inserting $\rho_B$ and $\sigma_B$ as $I_A \otimes \rho_B$ and $I_A \otimes \sigma_B$.

\subsection{Auxiliary Lemmas}

\begin{lemma}\label{variational optimization}
Let $Y \geq 0$ and let the optimization be over density matrices $\sigma$, with $Y \ll \sigma$ when $\mu<0$. Then
\begin{equation}\label{Schatten duality}
\opt_{\sigma}\Tr(Y\sigma^\mu) =
    \begin{cases}
        \left(\Tr Y^{\frac{1}{1-\mu}}\right)^{1-\mu}, & \quad \mu < 1, \\
        \|Y\|_\infty, & \quad \mu \geq 1,
    \end{cases}
\end{equation}
where $\opt$ is an infimum for $\mu<0$ and a supremum for $\mu\geq0$. If $Y > 0$ and $\mu < 1$, the optimizer is $\sigma_* = Y^{1/(1 - \mu)}/\Tr Y^{1/(1 - \mu)}$ and is unique.
\end{lemma}
\begin{proof}
    For $0 < \mu < 1$, H\"older's inequality gives
    \begin{equation}
        \Tr(Y \sigma^\mu) \leq \left(\Tr Y^{1/(1 - \mu)}\right)^{1 - \mu} (\Tr \sigma)^\mu,
    \end{equation}
    with equality at the stated $\sigma_*$. The case $\mu < 0$ is the corresponding reverse H\"older inequality; substituting the same $\sigma_*$ shows that the lower bound is attained. These are the standard equality case of noncommutative H\"older. For $\mu  = 0$ the formula follows by continuity. Finally, if $\mu \geq 1$, then $Y \leq \|Y\|_\infty I$ and $\Tr \sigma^\mu \leq 1$, so $\Tr(Y \sigma^\mu) \leq \|Y\|_\infty$. A rank-one state supported on a maximal eigenvector of $Y$ attains equality.
\end{proof}

\begin{lemma}\label{Schatten factorization lemma}
    Let $X$ be an arbitrary matrix and let $p,s,r > 0$ satisfy $1/r = 1/p + 1/s$. Then
    \begin{equation}\label{Schatten factorization}
        \inf_{\sigma \in \mathcal{S}_{++}} \left\|\sigma^{-1/s} X\right\|_p = \|X\|_r.
    \end{equation}
\end{lemma}
\begin{proof}
    Assume first that $X$ is invertible. For every faithful density matrix $\sigma$, generalized H\"older gives
    \begin{equation}
        \|X\|_r = \left\|\sigma^{1/s} \sigma^{-1/s} X\right\|_r \leq \left\|\sigma^{1/s}\right\|_s \left\|\sigma^{-1/s} X\right\|_p = \left\|\sigma^{-1/s} X\right\|_p,
    \end{equation}
    which proves the lower bound. Put $Y = XX^\dagger$ and $T = \Tr Y^{r/2}$. For $\sigma_* = Y^{r/2}/T$, the polar decomposition of $X$ gives
    \begin{equation}
        \sigma_*^{-1/s} X = T^{1/s} Y^{\frac{1}{2} - \frac{r}{2s}} U = T^{1/s} Y^{\frac{r}{2p}} U.
    \end{equation}
    Therefore,
    \begin{equation}
        \left\|\sigma_*^{-1/s} X\right\|_p^p = T^{p/s} \Tr Y^{r/2} = T^{p/r} = \|X\|_r^p,
    \end{equation}
    which proves \eqref{Schatten factorization}. Approximation gives the same identity for noninvertible $X$ whenever both sides are finite.
\end{proof}

\begin{lemma}\label{optimized Psi}
    Let $\alpha,z > 0$ and let $\omega$ be faithful. If $0 < \alpha < 1$, then
    \begin{equation}
        \sup_{\sigma \in \mathcal{S}} \Psi_{\alpha,z}(\omega\|\sigma) =1,
    \end{equation}
    whereas, if $\alpha > 1$, then
    \begin{equation}
        \inf_{\sigma \in \mathcal{S}_{++}} \Psi_{\alpha,z}(\omega\|\sigma) = 1.
    \end{equation}
    In both cases the optimizer is uniquely $\sigma = \omega$.
\end{lemma}
\begin{proof}
    Write
    \begin{equation}
        \Psi_{\alpha,z}(\omega\|\sigma) = \left\|\sigma^{(1 - \alpha)/2z} \omega^{\alpha/2z}\right\|_{2z}^{2z}.
    \end{equation}
    For $0 < \alpha < 1$, generalized H\"older yields
    \begin{equation}
    \left\|\sigma^{(1 - \alpha)/2z} \omega^{\alpha/2z}\right\|_{2z} \leq \left\|\sigma^{(1 - \alpha)/2z}\right\|_{2z/(1 - \alpha)} \left\|\omega^{\alpha/2z}\right\|_{2z/\alpha} = 1.
    \end{equation}
    Equality is attained at $\sigma = \omega$. If $\omega > 0$, the equality condition in H\"older's inequality implies that $\sigma = \omega$ is the unique optimizer. If $\alpha > 1$, apply Lemma \ref{Schatten factorization lemma} with $p = 2z$, $s = 2z/(\alpha - 1)$, $r = 2z/\alpha$, and $X = \omega^{\alpha/2z}$. The infimum of the norm is $1$, and the optimizer in the proof of Lemma \ref{Schatten factorization lemma} is precisely $\omega$.
\end{proof}

\begin{lemma}\label{renyi-parameter-uniqueness lemma}
Let $\rho\in\mathcal{S}(A)$ be a quantum state whose nonzero eigenvalues are not all equal. Then $t \mapsto H_t(\rho)$ is strictly decreasing on $[0,\infty]$. If $\rho$ is faithful, the same assertion holds on the extended real line $[-\infty,\infty]$.
\end{lemma}
\begin{proof}
    Let $p_1,\cdots,p_d$ bee the positive eigenvalues and put $F(t) = \log\sum_i p_i^t$. For the tilted distribution $w_i(t) = p_i^t/\sum_j p_j^t$,
    \begin{equation}
        F''(t) = \sum_i w_i(t) (\log p_i)^2 - \left(\sum_i w_i(t) \log p_i\right)^2 > 0,
    \end{equation}
    the inequality being strict because the $p_i$ are not all equal. Hence the secant slope $\big(F(t) - F(1)\big)/(t - 1)$ is strictly increasing in $t$. Since $F(1) = 0$ and $H_t = -\big(F(t) - F(1)\big)/(t - 1)$, the claim follows, including the endpoint orders by continuity.
\end{proof}

The following direct-sum rule is needed only in the duality section, where the parameters lie in $\mathcal{D}$ and the nondegenerate quantity $c = 1 - \lambda(1 - \alpha)$ is positive.
\begin{lemma}\cite[Proposition 11.1]{rubboli2026quantum}\label{classical register lemma}
Let $(\alpha,z,\lambda) \in \mathcal{D}$ and let
\begin{equation}
    \rho_{ABY} = \sum_y p_y \rho_{AB}^y \otimes \ketbra{y}{y}_Y.
\end{equation}
Set $\gamma = (1 - \alpha)/c$. Then
\begin{equation}
H_{\alpha,z}^{\lambda}(A|BY)_{\rho} = \frac{1}{\gamma} \log \sum_{y} p_y \exp\left(\gamma H_{\alpha,z}^{\lambda}(A|B)_{\rho^{y}}\right).
\end{equation}
\end{lemma}

We conclude this subsection with the following elementary optimization lemma, which will allow us, for a suitable family of classical-quantum states (see Proposition \ref{approximation lemma}), to express the optimized trace functional in \eqref{definition} asymptotically in an optimization-free form.

\begin{lemma}\label{optimization lemma}
    Let $K$ be compact, let $f,g : K \rightarrow \mathbb{R}$ be continuous, and suppose $x_0$ is the unique global optimum of $f$. Then
    \begin{equation}\label{optimization identity}
        \opt_{x \in K} \big(f(x) + \varepsilon g(x)\big) = f(x_0) + \varepsilon g(x_0) + o(\varepsilon) \qquad (\varepsilon \downarrow 0).
    \end{equation}
\end{lemma}
\begin{proof}
For each $\varepsilon > 0$, let $x_\varepsilon$ be an optimizer of the objective function in \eqref{optimization identity}, that is,
\begin{equation}
    \opt_{x \in K} \big(f(x) + \varepsilon g(x)\big) = f(x_\varepsilon) + \varepsilon g(x_\varepsilon).
\end{equation}
Such an optimizer exists by compactness of $K$ and continuity of the objective function. We first show that $x_\varepsilon \rightarrow x_0$. Let $\varepsilon_n \downarrow 0$, and let $x_{\varepsilon_{n_k}}$ be a subsequence of $x_{\varepsilon_n}$ converging to some $\bar{x} \in K$. In the maximization case, optimality gives, for every $x \in K$,
\begin{equation}
    f(x_{\varepsilon_{n_k}}) + \varepsilon_{n_k} g(x_{\varepsilon_{n_k}}) \geq f(x) + \varepsilon_{n_k} g(x).
\end{equation}
Passing to the limit and using continuity and boundedness of $g$, we obtain
\begin{equation}
    f(\bar{x}) \geq f(x), \qquad \text{for every} \ x \in K.
\end{equation}
Hence $\bar{x}$ maximizes $f$, and uniqueness implies $\bar{x} = x_0$. This means that $x_{\varepsilon_n} \rightarrow x_0$. Since this holds for every sequence $\varepsilon_n \downarrow 0$, it follows that $x_\varepsilon \rightarrow x_0$.

Next,
\begin{equation}
    f(x_0) + \varepsilon g(x_0) \leq \opt_{x \in K} \big(f(x) + \varepsilon g(x)\big) \leq f(x_0) + \varepsilon g(x_\varepsilon).
\end{equation}
dividing by $\epsilon$ after subtracting $f(x_0)$, and using $g(x_\varepsilon) \rightarrow g(x_0)$, gives
\begin{equation}
    \max_{x \in K} \big(f(x) + \varepsilon g(x)\big) = f(x_0) + \varepsilon g(x_0) + o(\varepsilon) \qquad (\varepsilon \downarrow 0).
\end{equation}
The minimum case follows by reversing the inequalities.
\end{proof}

\subsection{Conditional entropy of product states and pure states}
In this subsection, we derive explicit expressions for the conditional entropy of tensor-product states and pure states. For a tensor-product state, the conditional entropy reduces to the Rényi entropy of the marginal state on $A$, whereas for a pure state, it is given by the negative Rényi entropy of the marginal state on $B$. These reduction formulas will be used repeatedly in the subsequent sections.

The following lemma is an immediate consequence of the additivity of $H_{\alpha,z}^\lambda$ \cite[Theorem 4.2]{rubboli2026quantum}.
\begin{lemma}\label{product-state formula}
Assume $c > 0$, or assume $(\alpha,z,\lambda) \in \mathcal{D}$ and $c = 0$. If $\rho_{AB}=\rho_A\otimes\rho_B$, then
\begin{equation}\label{product-state}
H_{\alpha,z}^{\lambda}(A|B)_{\rho} = H_\alpha(\rho_A).
\end{equation}
In particular, if $\rho_A$ is pure, then $H_{\alpha,z}^\lambda(A|B)_\rho = 0$.
\end{lemma}
\begin{proof}
By tensor factorization when $\rho_{AB} = \rho_A \otimes \rho_B$, the trace functional in \eqref{definition} becomes $\Tr \rho_A^\alpha \Psi_{c,z}(\rho_B\|\sigma_B)$. Lemma \ref{optimized Psi} shows that the optimized second factor is $1$; taking the logarithm proves the claim. The $c = 0$ point in $\mathcal{D}$ is the continuous boundary $z = \alpha > 1$, and the same identity follows by the additivity of $H_{\alpha,z}^\lambda$ \cite[Theorem 4.2]{rubboli2026quantum}.
\end{proof}

Next, we compute $H_{\alpha,z}^\lambda$ for a pure state.
\begin{proposition}\label{pure-state formula}
Let $\rho_{AB}=\ketbra{\psi}{\psi}_{AB}$ be pure with faithful marginal $\rho_B$, and let $q,r$ be as in \eqref{parameters2}.
If $r < 1$, then
\begin{equation}\label{pure-state}
H_{\alpha,z}^{\lambda}(A|B)_{\rho} = -H_{\beta}(\rho_B), \qquad \beta = \frac{1 + q}{1 - r} = \frac{z+(1-\lambda)(1-\alpha)}{z-\lambda(1-\alpha)}.
\end{equation}
If $r \geq 1$, then
\begin{equation}\label{pure-state2}
    H_{\alpha,z}^{\lambda}(A|B)_{\rho} = \frac{z}{1 - \alpha} \log \big\|\rho_B^{1 + q}\big\|_\infty.
\end{equation}
\end{proposition}
\begin{proof}
Let $X_B = \rho_B^{q/2} \sigma_B^r \rho_B^{q/2}$.
Since $\rho_{AB}=\ketbra{\psi}{\psi}_{AB}$,
\begin{equation}
\rho_{AB}(I_A\otimes X_B)\rho_{AB} = \ketbra{\psi}{\psi} (I_A\otimes X_B) \ketbra{\psi}{\psi} = \Tr\big( \rho_{AB}(I_A\otimes X_B) \big)\rho_{AB}.
\end{equation}
Thus we have
\begin{align}
H_{\alpha,z}^{\lambda}(A|B)_{\rho} &= \frac{1}{1-\alpha} \opt_{\sigma_B}\log\Tr \Big( \Tr\!\big( \rho_{AB}(I_A\otimes X_B) \big)\rho_{AB} \Big)^z\\
&= \frac{z}{1-\alpha} \opt_{\sigma_B} \log \Tr\big(\rho_{AB}(I_A\otimes X_B) \big)\\
&= \frac{z}{1-\alpha} \log \opt_{\sigma_B} \Tr \left( \rho_B^{\frac{z+(1-\lambda)(1-\alpha)}{z}} \sigma_B^{\frac{\lambda(1-\alpha)}{z}} \right).
\end{align}
If $r < 1$, Lemma~\ref{variational optimization} gives
\begin{align}
H_{\alpha,z}^{\lambda}(A|B)_{\rho} &= \frac{z}{1-\alpha} \log \left( \Tr \left( \rho_B^{ \frac{z+(1-\lambda)(1-\alpha)}{z-\lambda(1-\alpha)} } \right) \right)^{ \frac{z-\lambda(1-\alpha)}{z} }\\
&= \frac{z-\lambda(1-\alpha)}{1-\alpha} \log \Tr \left( \rho_B^{ \frac{z+(1-\lambda)(1-\alpha)}{z-\lambda(1-\alpha)} } \right)\\ \label{pure-state-intermediate}
&= -\frac{1}{1-\beta}\log\Tr(\rho_B^\beta)\\
&= -H_\beta(\rho_B).
\end{align}
If $r \geq 1$, the second line of \eqref{Schatten duality} gives \eqref{pure-state2}.
\end{proof}


\section{Sharpness of Data-processing region}

In this section, we prove the exact converse to \cite[Theorem 5.2]{rubboli2026quantum}. The result is stated first for clarity and proved at the end of this section after the necessary restrictions have been derived.

\begin{theorem}[Sharp DPI region]\label{Sharp DPI region}
    Let $\alpha > 0$, $\alpha \neq 1$, $z > 0$, and $\lambda \in \mathbb{R}$. Then $H_{\alpha,z}^\lambda$ satisfies \eqref{data processing} for every finite-dimensional state and every admissible conditioned operation \eqref{conditional operation} if and only if
    \begin{equation}
        (\alpha,z,\lambda) \in \mathcal{D} = \mathcal{D}_1 \cup \mathcal{D}_2.
    \end{equation}
\end{theorem}

We begin with the following proposition, stating that DPI may hold only if $c := 1 - \lambda(1 - \alpha) \geq 0$.
\begin{proposition}\label{exclude c<0}
    If $c = 1 - \lambda(1 - \alpha) < 0$, then the DPI fails.
\end{proposition}
\begin{proof}
Observe first that DPI implies $H_{\alpha,z}^\lambda(X|Y)_P = 0$ for a trivial system $X$. Indeed, one may apply DPI both to the channel that discards $Y$ and to the channel that prepares $P_Y$ from a trivial conditioning system.

Let $P_Y = \sum_{y = 1}^d p_y \ketbra{y}{y}$ with $p_y > 0$ and $d \geq 2$. Then, if $c < 0$, it can readily be seen that $Q_{\alpha,z}^\lambda(X|Y)_P \geq \max_y p_y^c > 1$ since $0 < p_y < 1$. Consequently $H_{\alpha,z}^\lambda(X|Y)_P \neq 0$, which is a contradiction.
\end{proof}

Except for one genuine boundary family, which is treated directly in Proposition \ref{c=0 case}, we may consequently work with $c > 0$.

\subsection{Constraints on $z$}

Let $\rho_{XB} = \sum_{i} p_i \ketbra{i}{i}_X \otimes \omega_B^i$ be a classical--quantum state, where $p_i \geq 0$, $\sum_i p_i = 1$, and each $\omega_B^i$ is a density operator. Its $B$-marginal is $\rho_B = \sum_i p_i \omega_B^i$. Recall that for a fixed state $\sigma_B$,
\begin{align}
    Q_{\alpha,z}^\lambda(\rho_{XB}|\sigma_B) &:= \Tr\left(\rho_{XB}^{p/2} \left(I_X \otimes \rho_B^{q/2} \sigma_B^r \rho_B^{q/2}\right)\rho_B^{p/2}\right)^z,\\
    \label{block divergence}
    &= \sum_i p_i^\alpha \Tr\left(\left(\omega_B^i\right)^{p/2} \rho_B^{q/2} \sigma_B^r \rho_B^{q/2} \left(\omega_B^i\right)^{p/2}\right)^z
\end{align}
where $p = \alpha/z$, $q = (1 - \lambda)(1 - \alpha)/z$ and $r = \lambda(1 - \alpha)/z$. In particular, the right-hand side is independent of $\lambda$ for $\sigma_B = \rho_B$, and it can be written in terms of the $\alpha$-$z$ trace functional $\Psi_{\alpha,z}$:
\begin{equation}\label{reduced Q}
    Q_{\alpha,z}^\lambda(\rho_{XB}|\rho_B) = \sum_i p_i^\alpha \Psi_{\alpha,z}\left(\omega_B^i \| \rho_B\right).
\end{equation}

For a quantum state $\rho_{AB}$, we use the following notation:
\begin{equation}
    Q_{\alpha,z}^\lambda(A|B)_\rho := \opt_{\sigma_B} \Tr\left(\rho_{AB}^{\frac{\alpha}{2z}}\left(I_A\otimes\rho_B^{\frac{(1-\lambda)(1-\alpha)}{2z}}\sigma_B^{\frac{\lambda(1-\alpha)}{z}}\rho_B^{\frac{(1-\lambda)(1 - \alpha)}{2z}}\right)\rho_{AB}^{\frac{\alpha}{2z}}\right)^z,
\end{equation}
where $\opt$ is an infimum (resp.\ supremum) when $\lambda(1-\alpha)$ is negative (resp.\ non-negative).

\begin{proposition}\label{approximation lemma}
    Let $\alpha > 0$, $\alpha \neq 1$ and $c = 1 - \lambda(1 - \alpha) > 0$. Let $\omega$ be a faithful density matrix and let
    \begin{equation}
        \omega = \sum_{i = 1}^n \omega^i, \qquad \omega^i > 0,
    \end{equation}
    be a finite positive decomposition. Consider the one-parameter family of c-q states $\rho_{XB}^{(t)}$ defined by
    \begin{equation}\label{cq family}
        \rho_{XB}^{(t)} := (1 - t) \ketbra{0}{0}_X \otimes \omega + t \sum_{i = 1}^n \ketbra{i}{i}_X \otimes \omega^i, \qquad 0 < t < 1.
    \end{equation}
    Then
    \begin{equation}\label{approximate optimum}
        Q_{\alpha,z}^\lambda(X|B)_{\rho^{(t)}} = (1 - t)^\alpha + t^\alpha \sum_{i = 1}^n \Psi_{\alpha,z}\left(\omega^i \| \omega\right) + o(t^\alpha) \qquad (t \downarrow 0).
    \end{equation}
\end{proposition}
\begin{proof}
    Note first that the $B$-marginal of $\rho_{XB}^{(t)}$ is
    \begin{equation}
        \rho_B^{(t)} = (1 - t)\omega + t \sum_i \omega^i = \omega.
    \end{equation}
    For a fixed state $\sigma = \sigma_B$, using \eqref{block divergence} we obtain
    \begin{equation}\label{Q for cq family}
        Q_{\alpha,z}^\lambda\big(\rho_{XB}^{(t)}\big|\sigma\big) = (1 - t)^\alpha \Psi_{c,z}(\omega\|\sigma) + t^\alpha G(\sigma),
    \end{equation}
    where $G(\sigma) := \sum_i \Tr\left(\left(\omega^i\right)^{p/2} \omega^{q/2} \sigma^r \omega^{q/2} \left(\omega^i\right)^{p/2}\right)^z$.

    Suppose first that $c \neq 1$. By Lemma \ref{optimized Psi}, $\sigma = \omega$ is the unique optimizer of $\Psi_{c,z}(\omega\|\sigma)$: a maximum if $0 < c < 1$, and a minimum if $c > 1$. Factor
    \begin{equation}
        Q_{\alpha,z}^\lambda\big(\rho_{XB}^{(t)}\big|\sigma\big) = (1 - t)^\alpha \big(\Psi_{c,z}(\omega\|\sigma) + \varepsilon_t G(\sigma)\big), \qquad \varepsilon_t := \frac{t^\alpha}{(1 - t)^\alpha}.
    \end{equation}
    For $0 < c < 1$, the state space is compact and the two functions are continuous, so Lemma \ref{optimization lemma} applies directly. For $c > 1$, the negative powers make $\Psi_{c,z}(\omega\|\sigma)$ infinite on the boundary because $\omega$ is faithful. Comparing with the value at $\sigma = \omega$ shows that, for all sufficiently small $t$, every minimizer lies in one fixed compact faithful sublevel set. The same Lemma therefore applies on that set. We obtain
    \begin{equation}
        \opt_\sigma \big(\Psi_{c,z}(\omega\|\sigma) + \varepsilon_t G(\sigma)\big) = 1 + \varepsilon_t G(\omega) + o(\varepsilon_t).
    \end{equation}
    Multiplying by $(1 - t)^\alpha$ and using \eqref{reduced Q} proves \eqref{approximate optimum}.

    If $c = 1$, then $\lambda(1 - \alpha) = 0$ and the optimization is a supremum. For every faithful $\sigma$ on has $\sigma^0 = I$, so
    \begin{equation}
        Q_{\alpha,z}^\lambda\big(\rho_{XB}^{(t)}\big|\sigma\big) = (1 - t)^\alpha + t^\alpha G(\omega).
    \end{equation}
    With the support convention $\sigma^0 = \Gamma_\sigma$ for a singular state, th insertion $\Gamma \leq I$ cannot increase withrpositive block contribution. Thus a faithful state is optimal, and \eqref{approximate optimum} is exact without the $o(t^\alpha)$ term.
\end{proof}

\begin{theorem}\label{transfer}
    Let $c = 1 - \lambda(1 - \alpha) > 0$. If $H_{\alpha,z}^\lambda$ satisfies DPI, then $D_{\alpha,z}$ satisfies DPI under every quantum channel.
\end{theorem}
\begin{proof}
    Let $\omega_B$ be a faithful density matrix and let $\mathcal{N}_{B \rightarrow B'}$ be a quantum channel. Without loss of generality, we may assume that all output states of $\mathcal{N}$ are faithful. We first derive an inequality associated with an arbitrary positive decomposition.
    
    Let $\omega = \sum_{i = 1}^n \omega^i$ be a finite positive decomposition where $\omega^i > 0$ for all $i$. For $0 < t < 1$, define the c-q state $\rho_{XB}^{(t)}$ as in \eqref{cq family}. By Lemma \ref{approximation lemma},
    \begin{equation}\label{approximation identity appl}
        Q_{\alpha,z}^\lambda(X|B)_{\rho^{(t)}} = (1 - t)^\alpha + t^\alpha \sum_{i = 1}^n \Psi_{\alpha,z}(\omega^i\|\omega) + o(t^\alpha).
    \end{equation}
    On the other hand, applying $\mathcal{N}$ to $\rho^{(t)}$ gives
    \begin{equation}
        \mathcal{N}(\rho^{(t)}) = (1 - t) \ketbra{0}{0}_X \otimes \mathcal{N}(\omega) + t \sum_i \ketbra{i}{i}_X \otimes \mathcal{N}(\omega^i).
    \end{equation}
    Sine $\mathcal{N}(\omega) = \sum_i \mathcal{N}(\omega^i)$, this is again a state of the same form. Lemma \ref{approximation lemma} therefore gives
    \begin{equation}\label{approximation identity appl2}
        Q_{\alpha,z}^\lambda(X|B')_{\mathcal{N}(\rho^{(t)})} = (1 - t)^\alpha + t^\alpha \sum_{i = 1}^n \Psi_{\alpha,z}\big(\mathcal{N}(\omega^i)\|\mathcal{N}(\omega)\big) + o(t^\alpha).
    \end{equation}

    We compare \eqref{approximation identity appl} and \eqref{approximation identity appl2} using conditional DPI
    \begin{equation}
        H_{\alpha,z}^\lambda(X|B)_{\rho^{(t)}} \leq H_{\alpha,z}^\lambda(X|B')_{\mathcal{N}(\rho^{(t)})}.
    \end{equation}
    In the case $0 < \alpha < 1$, this inequality is equivalent to
    \begin{equation}
        Q_{\alpha,z}^\lambda(X|B)_{\rho^{(t)}} \leq Q_{\alpha,z}^\lambda(X|B')_{\mathcal{N}(\rho^{(t)})}.
    \end{equation}
    Substituting \eqref{approximation identity appl} and \eqref{approximation identity appl2}, canceling the common term $(1 - t)^\alpha$, dividng by $t^\alpha$, and letting $t \downarrow 0$, we obtain
    \begin{equation}
        \sum_{i = 1}^n \Psi_{\alpha,z}(\omega^i\|\omega) \leq \sum_{i = 1}^n \Psi_{\alpha,z}\big(\mathcal{N}(\omega^i)\|\mathcal{N}(\omega)\big).
    \end{equation}
    In the case $\alpha > 1$, the conditional DPI is equivalent to the reverse inequality
    \begin{equation}
        Q_{\alpha,z}^\lambda(X|B)_{\rho^{(t)}} \geq Q_{\alpha,z}^\lambda(X|B')_{\mathcal{N}(\rho^{(t)})}.
    \end{equation}
    The same coefficient comparison gives
    \begin{equation}
        \sum_{i = 1}^n \Psi_{\alpha,z}(\omega^i\|\omega) \geq \sum_{i = 1}^n \Psi_{\alpha,z}\big(\mathcal{N}(\omega^i)\|\mathcal{N}(\omega)\big).
    \end{equation}
    Thus, for every finite positive decomposition $\omega = \sum_i \omega^i$,
    \begin{equation}\label{positive decomposition inequality}
    \sum_i \Psi_{\alpha,z}(\omega^i\|\omega)
        \begin{cases}
            \leq \sum\limits_i \Psi_{\alpha,z}\big(\mathcal{N}(\omega^i)\|\mathcal{N}(\omega)\big), & 0 < \alpha < 1,\\
            \geq \sum\limits_i \Psi_{\alpha,z}\big(\mathcal{N}(\omega^i)\|\mathcal{N}(\omega)\big), & \alpha > 1.
        \end{cases}
    \end{equation}

    Now let $\sigma_B$ be a faithful density matrix such that $\sigma_B < \omega_B$. For each positive integer $n$, we consider the decomposition
    \begin{equation}
        \omega = \sigma + \underbrace{\frac{\omega - \sigma}{n} + \cdots + \frac{\omega - \sigma}{n}}_{n \text{ copies}}.
    \end{equation}
    In the case $\alpha > 1$, applying \eqref{positive decomposition inequality} gives
    \begin{equation}\label{scaled inequlaity}
        \Psi_{\alpha,z}(\sigma\|\omega) + n \Psi_{\alpha,z}\left(\frac{\omega - \sigma}{n}\Big\|\omega\right) \geq \Psi_{\alpha,z}\big(\mathcal{N}(\sigma)\|\mathcal{N}(\omega)\big) + n \Psi_{\alpha,z}\left(\frac{\mathcal{N}(\omega - \sigma)}{n}\Big\|\mathcal{N}(\omega)\right).
    \end{equation}
    The trace functional is homogeneous of degree $\alpha$ in its first argument. Therefore,
    \begin{equation}
        n \Psi_{\alpha,z}\left(\frac{\omega - \sigma}{n}\Big\|\omega\right) = n^{1 - \alpha} \Psi_{\alpha,z}(\omega - \sigma\|\omega),
    \end{equation}
    and similarly after applying $\mathcal{N}$. Hence \eqref{scaled inequlaity} becomes
    \begin{equation}
        \Psi_{\alpha,z}(\sigma\|\omega) + n^{1 - \alpha} \Psi_{\alpha,z}(\omega - \sigma\|\omega) \geq \Psi_{\alpha,z}\big(\mathcal{N}(\sigma)\|\mathcal{N}(\omega)\big) + n^{1 - \alpha} \Psi_{\alpha,z}\big(\mathcal{N}(\omega - \sigma)\|\mathcal{N}(\omega)\big).
    \end{equation}
    Since $\alpha > 1$, we have $m^{1 - \alpha} \longrightarrow 0$. Letting $m \rightarrow \infty$ in the preceding inequality, we obtain
    \begin{equation}\label{pre DPI}
        \Psi_{\alpha,z}(\sigma\|\omega) \geq \Psi_{\alpha,z}\big(\mathcal{N}(\sigma)\|\mathcal{N}(\omega)\big).
    \end{equation}

    In the case $0 < \alpha < 1$, we use the decomposition
    \begin{equation}
        \omega = \underbrace{\frac{\sigma}{n} + \cdots + \frac{\sigma}{n}}_{n \text{ copies}} + (\omega - \sigma).
    \end{equation}
    Applying \eqref{positive decomposition inequality} and using homogeneity, we obtain
    \begin{equation}
        \Psi_{\alpha,z}(\sigma\|\omega) + n^{\alpha - 1} \Psi_{\alpha,z}(\omega - \sigma\|\omega) \leq \Psi_{\alpha,z}\big(\mathcal{N}(\sigma)\|\mathcal{N}(\omega)\big) + n^{\alpha - 1} \Psi_{\alpha,z}\big(\mathcal{N}(\omega - \sigma)\|\mathcal{N}(\omega)\big).
    \end{equation}
    Since $0 < \alpha < 1$, we have $m^{\alpha - 1} \longrightarrow 0$. Letting $m \rightarrow \infty$, we obtain
    \begin{equation}\label{pre DPI2}
        \Psi_{\alpha,z}(\sigma\|\omega) \leq \Psi_{\alpha,z}\big(\mathcal{N}(\sigma)\|\mathcal{N}(\omega)\big).
    \end{equation}
    For arbitrary faithful states $\sigma,\omega$, one can choose $s$ so large that $\sigma/s < \omega$. Repeating the argument above with $\sigma/s$ in place of $\sigma$, and using homogeneity to cancel the factor $s^{-\alpha}$, we obtain the same inequalities \eqref{pre DPI} and \eqref{pre DPI2} which are equivalent to DPI for $D_{\alpha,z}$ after accounting for the sign of $\alpha - 1$.

    To remove the temporary strict-positivity assumption on the channel, replace $\mathcal{N}$ by
    \begin{equation}
        \mathcal{N}_\delta(X) = (1 - \delta) \mathcal{N}(X) + \delta \Tr(X) \tau,
    \end{equation}
    where $\tau$ is faithful, and let $\delta \downarrow 0$. Positivity of the channel implies that its outputs at faithful states have the same support; continuity on that support yields the limit. Nonfaithful pairs follow by adding a faithful state an then taking the regularization parameter to zero. This also covers the extended-valued case when $\alpha > 1$.
\end{proof}

\subsection{Constraints on $\lambda$}

\begin{proposition}[The upper $\lambda$ bound]\label{upper lambda bound}
    Let $\alpha > 0$, $\alpha \neq 1$ and $c,z > 0$. Let $H_{\alpha,z}^\lambda$ satisfy the DPI. If $0 < \alpha < 1$, then $0 < \beta \leq 1$; if $\alpha >1$, then $\beta \geq 1$. In either case, $\lambda \leq 1$.
\end{proposition}
\begin{proof}
Put $\beta := \alpha/(1 - \lambda(1 - \alpha)) > 0$ and define on the probability simplex
\begin{equation}
    p \longmapsto \|p\|_\alpha^\beta := \left(\sum_x p_x^\alpha\right)^{\beta/\alpha}.
\end{equation}
Let $P_{XY}$ be a faithful classical joint distribution and write $p_{xy} = p_y p_{x|y}$. If $\lambda \leq 0$, then the assertions are immediate. Hence we assume that $\lambda > 0$. Then, restricting the optimization in the definition of $H_{\alpha,z}^{\lambda}(X|Y)_P$ to diagonal reference states and applying Lemma~\ref{variational optimization} gives
\begin{equation}\label{clssical lower bound}
    \frac{\alpha}{(1-\alpha)\beta} \log\sum_y p_y \left(\sum_x p_{x|y}^\alpha\right)^{\beta/\alpha} \leq H_{\alpha,z}^\lambda(X|Y)_P.
\end{equation}
Discarding $Y$ maps the collection of conditional distributions $p^{(y)} = (p_{x|y})_x$ to their mixture $p_X = \sum_y p_y p^{(y)}$. Combined with \eqref{clssical lower bound}, DPI under discarding the conditioning system $Y$ implies
\begin{equation}
\begin{cases}
    \left\|\sum_y p_y p^{(y)}\right\|_\alpha^\beta \geq \sum_y p_y \left\|p^{(y)}\right\|_\alpha^\beta, \quad & 0 < \alpha < 1,\\[7pt]
    \left\|\sum_y p_y p^{(y)}\right\|_\alpha^\beta \leq \sum_y p_y \left\|p^{(y)}\right\|_\alpha^\beta, \quad & \alpha > 1.
\end{cases}
\end{equation}
Thus $p \mapsto \|p\|_\alpha^\beta$ must be concave for $0 < \alpha < 1$ and convex for $\alpha > 1$.

The function $p \mapsto \|p\|_\alpha^\beta$ is smooth on the positive orthant. Let
\begin{equation}
    N(p) := \sum_i p_i^\alpha.
\end{equation}
For a tangent vector $h$ satisfying $\sum_i h_i = 0$, put
\begin{equation}
    A(p,h) := \sum_i p_i^{\alpha - 1} h_i, \qquad B(p,h) := \sum_i p_i^{\alpha -2} h_i^2.
\end{equation}
By direct differentiation, we compute the second derivative of $p \mapsto \|p\|_\alpha^\beta$ in the direction $h$:
\begin{equation}\label{second directional derivatve}
    D_h^2 \|p\|_\alpha^\beta = \beta N(p)^{\beta/\alpha - 2}\Big((\alpha -1) N(p) B(p,h) + (\beta - \alpha) A(p,h)^2\Big).
\end{equation}

For $n \geq 1$, consider the probability vector and tangent direction
\begin{equation}
    p^{(n)} = \left(\frac{1}{2},\underbrace{\frac{1}{2n},...,\frac{1}{2n}}_{n\text{ entries}}\right), \qquad h^{(n)} = \left(-\frac{1}{2},\underbrace{\frac{1}{2n},...,\frac{1}{2n}}_{n\text{ entries}}\right).
\end{equation}
A direct calculation gives
\begin{equation}
    A_n := A(p^{(n)},h^{(n)}) = 2^{-\alpha}(n^{1 - \alpha} - 1), \qquad B_n := B(p^{(n)},h^{(n)}) = 2^{-\alpha}(1 + n^{1 - \alpha}).
\end{equation}
Note also that $N(p^{(n)}) = B_n$ which follows term by term:
\begin{equation}
    \left(\frac{1}{2}\right)^{\alpha - 2} \left(\frac{1}{2}\right)^2 = 2^{-\alpha}, \qquad n \left(\frac{1}{2n}\right)^{\alpha - 2} \left(\frac{1}{2n}\right)^2 = 2^{-\alpha} n^{1 - \alpha}.
\end{equation}
Moreover,
\begin{equation}\label{lim}
    \frac{A_n^2}{B_n^2} = \left(\frac{n^{1 - \alpha} -1}{n^{1 - \alpha} + 1}\right)^2 \longrightarrow 1, \qquad \text{as} \quad n \rightarrow \infty.
\end{equation}
Substituting $N(p^{(n)}) = B_n$ into \eqref{second directional derivatve} and factoring out the strictly positive quantity $\beta B_n^{\beta/\alpha}$ gives
\begin{equation}
    D_{h^{(n)}}^2 \left\|p^{(n)}\right\|_\alpha^\beta = \beta B_n^{\beta/\alpha} \left((\alpha - 1) + (\beta - \alpha) \frac{A_n^2}{B_n^2}\right).
\end{equation}
The bracket converges to $\beta - 1$. Concavity therefore excludes $\beta > 1$ when $0 < \alpha < 1$, while convexity excludes $\beta < 1$ when $\alpha > 1$. In both cases, it follows from the definition of $\beta$ that $\lambda \leq 1$.
\end{proof}

\begin{proposition}[The lower $\lambda$ bound for $0 < \alpha < 1$]\label{lower lambda bound}
   Assume $c > 0$, $0 < \alpha < 1$, and DPI. Then
   \begin{equation}
       \lambda \geq 1 - \frac{z}{1 - \alpha}.
   \end{equation}
\end{proposition}
\begin{proof}
If $\lambda \geq 0$, the desired lower bound is automatic. It remains to treat $\lambda < 0$. Put $p := \alpha/z$, $q := (1 - \lambda)(1 - \alpha)/z$ and $r := \lambda(1 - \alpha)/z$, and $s := z$.

Let $\omega$ be a faithful density matrix and let $\sigma$ be positive with $0 < \sigma < \omega$. For each $n$, consider the c-q state
\begin{equation}
    \rho_{XB}^{(n)} := \ketbra{0}{0}_X \otimes (\omega - \sigma) + \sum_{j = 1}^n \ketbra{j}{j}_X \otimes \frac{\sigma}{n}.
\end{equation}
Its $B$-marginal is $\omega$. Since $r < 0$, the optimization in the conditional entropy is infimum. Block diagonality therefore gives
\begin{equation}
    Q_{\alpha,z}^\lambda(X|B)_{\rho^{(n)}} = \inf_{\tau \in \mathcal{S}_{++}} \Big\{\Phi_{p,q,r,s}(\omega - \sigma,\omega,\tau) + n^{1 - \alpha} \Phi_{p,q,r,s}(\sigma,\omega,\tau)\Big\}.
\end{equation}
Because $0 < \alpha < 1$, we have $n^{1 - \alpha} \rightarrow \infty$, and hence
\begin{equation}\label{limiting optimization}
    \lim_{n \rightarrow \infty} n^{\alpha - 1} Q_{\alpha,z}^\lambda(X|B)_{\rho^{(n)}} = \inf_{\tau \in \mathcal{S}_{++}} \Phi_{p,q,r,s}(\sigma,\omega,\tau).
\end{equation}
Indeed, after multiplying the preceding variational expression by $n^{\alpha - 1}$, the resulting objective is bounded below by $\Phi_{p,q,r,s}(\sigma,\omega,\tau)$. This gives the lower bound. Conversely, evaluating the objective at an arbitrarily good faithful state for the infimum on the right-hand side gives the matching upper bound, since the contribution involving $\omega - \sigma$ is multiplied by $n^{\alpha - 1} \rightarrow 0$.

Let $\mathcal{N}_{B \rightarrow B'}$ be a quantum channel. Applying $\mathcal{N}$ to $\rho_{XB}^{(n)}$ produces the same construction with $\sigma,\omega$ replaced by $\mathcal{N}(\sigma),\mathcal{N}(\omega)$. Since $0 < \alpha < 1$, conditional DPI is equivalent to
\begin{equation}
    Q_{\alpha,z}^\lambda(X|B)_{\rho^{(n)}} \leq Q_{\alpha,z}^\lambda(X|B')_{\mathcal{N}(\rho^{(n)})}
\end{equation}
Multiplying by $n^{\alpha - 1}$ and using \eqref{limiting optimization} yields
\begin{equation}\label{optimized ineq}
    \inf_{\tau_{B}} \Phi_{p,q,r,s}(\sigma,\omega,\tau) \leq \inf_{\tau_{B'}} \Phi_{p,q,r,s}\big(\mathcal{N}(\sigma),\mathcal{N}(\omega),\tau_{B'}\big).
\end{equation}
The same harmless depolarizing regularization used in Theorem \ref{transfer} makes all matrices faithful.

Set
\begin{equation}
    t = \frac{z}{1 - rz} = \frac{z}{c} > 0.
\end{equation}
For $X = \omega^{q/2} \sigma^{p/2}$,
\begin{equation}
    \Phi_{p,q,r,s}(\sigma,\omega,\tau) = \big\|\tau^{r/2} X\big\|_{2s}^{2s}.
\end{equation}
Applying Lemma \ref{Schatten factorization lemma}, we obtain the exact identity
\begin{align}
    \inf_{\tau \in \mathcal{S}_{++}} \Phi_{p,q,r,s}(\sigma,\omega,\tau) &= \big\|\omega^{q/2} \sigma^{p/2}\big\|_{2s}^{2t}\\
    &= \Psi_{pt,t}(\sigma\|\omega)^{z/t} = \Psi_{pt,t}(\sigma\|\omega)^c.
\end{align}
Since $c > 0$, \eqref{optimized ineq} implies DPI for $D_{pt,t}$. Zhang's exact characterization of data-processing for the $\alpha$-$z$ divergence \cite{zhang20_alphaz} now implies $t \geq 1 - pt$, which is equivalent to $\lambda \geq 1 - /(1 - \alpha)$, as desired.
\end{proof}

\begin{proposition}[The lower $\lambda$ bound for $\alpha > 1$]\label{lower lambda bound for alpha>1}
    Assume $c \geq 0$, $\alpha > 1$, and DPI. Then
    \begin{equation}
        \lambda \geq 1 + \frac{z}{1 - \alpha}.
    \end{equation}
\end{proposition}
\begin{proof}
Let $\psi_P$ denote a bipartite pure state with Schmidt vector $P = (p_1,\cdots,p_d)$, i.e.,
\begin{equation}
    \ket{\psi_P}_{AB} := \sum_{i = 1}^d \sqrt{p_i} \ket{i}_A \ket{i}_B, \qquad p_i > 0.
\end{equation}

Towards a contradiction, assume $\lambda < 1 + z/(1 - \alpha)$, or equivalently, $q < -1$. If $r < 1$, then $\gamma := (1 + q)/(1 - r) < 0$. For the uniform vector $u_d$, $H_\gamma(u_d) = \log d$, while every faithful nonuniform $P$ satisfies $H_\gamma(P) > \log d$. Thus,
\begin{equation}\label{strict decrease}
    H_{\alpha,z}^\lambda(A|B)_{\psi_P} = -H_\gamma(P) < -H_\gamma(u_d) = H_{\alpha,z}^\lambda(A|B)_{\psi_{u_d}}.
\end{equation}
The same inequality holds when $r \geq 1$.

Nielsen's theorem \cite{nielsen1999conditions} states that $\ket{\psi_x}$ can be converted deterministically to $\ket{\psi_y}$ by LOCC if and only if $x$ is majorized by $y$. Since $u_d$ is majorized by $P$, the maximally entangled state van be converted to $\psi_P$. The strict decrease in \eqref{strict decrease} contradicts DPI. Hence $q \geq -1$, which is equivalent to $\lambda \geq 1 + z/(1 - \alpha)$.
\end{proof}

\begin{proposition}[The degenerate case $c = 0$]\label{c=0 case}
    Suppose $c = 1 - \lambda(1 - \alpha) = 0$ and $H_{\alpha,z}^\lambda$ satisfies DPI. Then
    \begin{equation}
        \alpha > 1, \qquad z= \alpha, \qquad \lambda = \frac{1}{1 - \alpha}.
    \end{equation}
    In particular, this case is exactly the $z = \alpha$ endpoint of the lower $\lambda$ boundary of $\mathcal{D}_2$.
\end{proposition}
\begin{proof}
    When $c = 0$, \eqref{parameters2} gives $q = -p$, $r = 1/z$, $p = \alpha/z$. For the state \eqref{cq family}, the first term in \eqref{Q for cq family} is exactly $(1 - t)^\alpha$, independently of the optimizer. Consequently,
    \begin{equation}\label{Q in the proof of degenerate case}
        Q_{\alpha,z}^\lambda(X|B)_{\rho^{(t)}} = (1 - t)^\alpha + t^\alpha \sup_{\tau\in\mathcal{S}}G(\tau).
    \end{equation}
    For every $0 < \sigma < \omega$, we can isolate the functional
    \begin{equation}
        (\sigma,\omega)\longmapsto\sup_{\tau\in\mathcal{S}} \Phi_{p,-p,1/z,z}(\sigma,\omega,\tau) = \big\|\omega^{-p/2}\sigma^{p/2}\big\|_\infty^{2z}
    \end{equation}
    as in the proof of Theorem~\ref{transfer}. Indeed, for $\alpha > 1$ use $\omega = \sigma n(\omega - \sigma)/n$, while for $0 < \alpha < 1$ use $\omega = n(\sigma/n) + (\omega - \sigma)$ and multiply by $n^{\alpha - 1}$. Equation \eqref{Q in the proof of degenerate case} and DPI imply for every channel $\mathcal{N}$,
    \begin{equation}\label{positive decomposition inequality2}
    \big\|\omega^{-p/2}\sigma^{p/2}\big\|_\infty
        \begin{cases}
            \leq \big\|\mathcal{N}(\omega)^{-p/2}\mathcal{N}(\sigma)^{p/2}\big\|_\infty, & 0 < \alpha < 1,\\[5pt]
            \geq \big\|\mathcal{N}(\omega)^{-p/2}\mathcal{N}(\sigma)^{p/2}\big\|_\infty, & \alpha > 1.
        \end{cases}
    \end{equation}

    If $0 < \alpha < 1$, we choose commuting diagonal $0 < \sigma < \omega$ for which the eigenvalue ratios $\sigma_i/\omega_i$ are not all equal and let $\mathcal{N} = \Tr$. Then
    \begin{equation}
        \big\|(\Tr\omega)^{-p/2}(\Tr\sigma)^{p/2}\big\|_\infty = \left(\frac{\Tr\sigma}{\Tr\omega}\right)^{p/2} < \max_i \left(\frac{\sigma_i}{\omega_i}\right)^{p/2} = \big\|\omega^{-p/2}\sigma^{p/2}\big\|_\infty,
    \end{equation}
    contradicting the first line of \eqref{positive decomposition inequality2}. Hence $\alpha > 1$.

    Since $c = 0$ is equivalent to $\lambda = 1/(1 - \alpha)$, applying Proposition \ref{lower lambda bound for alpha>1} gives $z \geq \alpha$, or $p \leq 1$. Toward a contradiction, suppose $p < 1$. The L\"owner--Heinz theorem states that $t\mapsto t^a$ is operator monotone on positive matrices exactly for $0 \leq a \leq 1$ \cite[Chapter V]{bhatia1997graduate}. Since $1/p > 1$, there exist positive definite $A \leq B$ such that $A^{1/p} \nleq B^{1/p}$. Put $\sigma_0 = A^{1/p}$ and $\omega_0 = B^{1/p}$. Then
    \begin{equation}
        \big\|\omega_0^{-p/2}\sigma_0^{p/2}\big\|_\infty = \big\|B^{-1/2}A^{1/2}\big\|_\infty \leq 1,
    \end{equation}
    whereas there is a unit vector $v$ with $\langle v,\sigma_0 v\rangle > \langle v,\omega_0 v\rangle$. Choose $0 < k < 1$ so small that $k\sigma_0 < \omega_0$, and consider the binary quantum-to-classical measurement channel $\mathcal{N}$ with first effect $P_v = \ketbra{v}{v}$:
    \begin{equation}
        \mathcal{N}(X) = \Tr(P_v X)\ketbra{0}{0} + \Tr\big((I - P_v)X\big)\ketbra{1}{1}.
    \end{equation}
    For inputs $\sigma = k \sigma_0$ and $\omega = \omega_0$, the first output coordinates of $\mathcal{N}(\sigma)$ and $\mathcal{N}(\omega)$ are $k\langle v,\sigma_0 v\rangle$ and $\langle v,\omega_0 v\rangle$, respectively. Then
    \begin{align}
        \big\|\omega^{-p/2}\sigma^{p/2}\big\|_\infty &= k^{p/2}\big\|\omega_0^{-p/2}\sigma_0^{p/2}\big\|_\infty = k^{p/2}\big\|B^{-1/2}A^{1/2}\big\|_\infty\\
        &\leq k^{p/2} < \left(\frac{k\langle v,\sigma_0 v\rangle}{\langle v,\omega_0 v\rangle}\right)^{p/2} \leq \big\|\mathcal{N}(\omega)^{-p/2}\mathcal{N}(\sigma)^{p/2}\big\|_\infty.
    \end{align}
    This contradicts the second line of \eqref{positive decomposition inequality2}. Hence $p = 1$, i.e. $z = \alpha$.
\end{proof}

\begin{proof}[Proof of Theorem \ref{Sharp DPI region}]
    Sufficiency is \cite[Theorem 5.2]{rubboli2026quantum}. For necessity, Proposition \ref{exclude c<0} excludes $c = 1 - \lambda(1 - \alpha) < 0$. If $c = 0$, Proposition \ref{c=0 case} gives precisely a point on the lower boundary of $\mathcal{D}_2$. If $c > 0$, Theorem \ref{transfer} supplies the stated $z$ bounds, Proposition \ref{upper lambda bound} gives $\lambda \leq 1$, and Propositions \ref{lower lambda bound} and \ref{lower lambda bound for alpha>1} give the appropriate lower bound. These conditions are exactly \eqref{cconstraints1} and \eqref{constraints2}.
\end{proof}


\section{Coarse-graining and the club-sandwiched family}

Let $X$ be a classical register with finite alphabet $\mathcal{X}$, and let $B$ be a finite-dimensional quantum system. A classical--quantum state on $XB$ can be written as
\begin{equation}
\rho_{XB} = \sum_{x\in\mathcal{X}} |x\rangle\langle x|_X \otimes T_x, \qquad T_x \geq 0, \qquad \rho_B = \sum_{x\in\mathcal{X}} T_x.
\end{equation}
For an arbitrary deterministic function $f:\mathcal{X}\to\mathcal{Z}$, let $R_f$ denote the associated classical channel defined by
\begin{equation}
R_f\left(\lvert x\rangle\langle x\rvert_X\right) = \lvert f(x)\rangle\langle f(x)\rvert_Z.
\end{equation}
We say that $H^\lambda_{\alpha,z}$ satisfies universal coarse-graining if, for every finite classical--quantum state $\rho_{XB}$ and every deterministic function $f$,
\begin{equation}
H^\lambda_{\alpha,z}(X|B)_\rho \geq H^\lambda_{\alpha,z}(Z|B)_{(R_f\otimes\operatorname{id}_B)(\rho)}. \label{coarse-graining-inequality}
\end{equation}

We now introduce the main result of this section, which shows that universal monotonicity under deterministic coarse-graining characterizes the club-sandwiched surface within the data-processing region.
\begin{theorem}\label{coarse graining thm}
Let $(\alpha,z,\lambda)\in D$. Then the following are equivalent:
\begin{enumerate}
\item[(i)] $z=\alpha$;
\item[(ii)] $H^\lambda_{\alpha,z}$ satisfies universal coarse-graining \eqref{coarse-graining-inequality}.
\end{enumerate}
\end{theorem}

For the necessity direction, we first derive an auxiliary consequence of universal coarse-graining that reduces the problem to a matrix inequality independent of $\lambda$. For a faithful state $\tau$ and $A,C\geq 0$, define
\begin{equation}
\Delta_\tau(A,C):=\Psi_{\alpha,z}(A+C\Vert\tau)-\Psi_{\alpha,z}(A\Vert\tau)-\Psi_{\alpha,z}(C\Vert\tau).
\end{equation}

\begin{lemma}\label{coarse graining criterion}
Let $(\alpha,z,\lambda)\in D$ with $z\neq\alpha$. If $H^\lambda_{\alpha,z}$ satisfies universal coarse-graining, then for every faithful state $\tau$ and every $A,C\geq 0$,
\begin{equation}
(\alpha-1)\Delta_\tau(A,C)\geq 0.
\end{equation}
\end{lemma}
\begin{proof}
For $\eta>0$, let $A_\eta = A+\eta I$ and $C_\eta = C+\eta I$. Then $A_\eta,C_\eta>0$. Since $\tau>0$, there exists a sufficiently small $\varepsilon>0$ such that $R = \tau-\varepsilon(A_\eta+C_\eta) > 0$, and hence $\tau = R+\varepsilon A_{\eta}+\varepsilon C_{\eta}$. Since $(\alpha,z,\lambda)\in D$ and $z\neq\alpha$, we have $c:=1-\lambda(1-\alpha)>0$, so Proposition~\ref{approximation lemma} applies. For $0<t<1$, consider the classical--quantum state
\begin{equation}
\rho^{\mathrm f}_{XB}(t) = (1-t)\lvert 0\rangle\langle 0\rvert_X \otimes \tau + t\lvert 1\rangle\langle 1\rvert_X \otimes R + t\lvert 2\rangle\langle 2\rvert_X \otimes \varepsilon A_\eta + t\lvert 3\rangle\langle 3\rvert_X \otimes \varepsilon C_\eta.
\end{equation}
Let $f:\{0,1,2,3\}\to\{0,1,2\}$ be given by $f(0)=0$, $f(1)=1$, and $f(2)=f(3)=2$. The corresponding coarse-grained state is
\begin{equation}
\rho^{\mathrm c}_{ZB}(t) = (R_f \otimes \operatorname{id}_B)(\rho^{\mathrm f}_{XB}(t)) = (1-t)\lvert 0\rangle\langle 0\rvert_Z \otimes \tau + t\lvert 1\rangle\langle 1\rvert_Z \otimes R + t\lvert 2\rangle\langle 2\rvert_Z \otimes \varepsilon(A_\eta+C_\eta).
\end{equation}
Both states have marginal $\tau$ on $B$. Proposition~\ref{approximation lemma}, together with the homogeneity of $\Psi_{\alpha,z}$, gives
\begin{align}
Q^\lambda_{\alpha,z}(X|B)_{\rho^{\mathrm f}(t)} &= (1-t)^\alpha + t^\alpha\left(\Psi_{\alpha,z}(R\Vert\tau) + \varepsilon^\alpha\Psi_{\alpha,z}(A_\eta\Vert\tau) + \varepsilon^\alpha\Psi_{\alpha,z}(C_\eta\Vert\tau)\right) + o(t^\alpha),\\
Q^\lambda_{\alpha,z}(Z|B)_{\rho^{\mathrm c}(t)} &= (1-t)^\alpha + t^\alpha\left(\Psi_{\alpha,z}(R\Vert\tau) + \varepsilon^\alpha\Psi_{\alpha,z}(A_\eta + C_\eta\Vert\tau)\right)+o(t^\alpha).
\end{align}
Therefore,
\begin{equation}
Q^\lambda_{\alpha,z}(X|B)_{\rho^{\mathrm f}(t)} - Q^\lambda_{\alpha,z}(Z|B)_{\rho^{\mathrm c}(t)} = -t^\alpha\varepsilon^\alpha\Delta_\tau(A_\eta,C_\eta) + o(t^\alpha).
\end{equation}
Universal coarse-graining implies
\begin{equation}
H^\lambda_{\alpha,z}(X|B)_{\rho^{\mathrm f}(t)}\geq H^\lambda_{\alpha,z}(Z|B)_{\rho^{\mathrm c}(t)}.
\end{equation}
If $0<\alpha<1$, then $Q^\lambda_{\alpha,z}(X|B)_{\rho^{\mathrm f}(t)}\geq Q^\lambda_{\alpha,z}(Z|B)_{\rho^{\mathrm c}(t)}$. Dividing by $t^\alpha$ and letting $t\downarrow 0$ gives $\Delta_\tau(A_\eta,C_\eta)\leq 0$. If $\alpha>1$, the inequality between the two $Q$-quantities is reversed, and the same argument gives $\Delta_\tau(A_\eta,C_\eta)\geq 0$. Hence,
\begin{equation}
(\alpha-1)\Delta_\tau(A_\eta,C_\eta)\geq 0.
\end{equation}
Since $\Psi_{\alpha,z}(T\Vert\tau)$ is continuous in $T\geq 0$, we have $\Delta_\tau(A_\eta,C_\eta)\to\Delta_\tau(A,C)$ as $\eta\downarrow 0$. Therefore,
\begin{equation}
(\alpha-1)\Delta_\tau(A,C)\geq 0.
\end{equation}
The proof is complete.
\end{proof}

We are now ready to prove Theorem~\ref{coarse graining thm}.
\begin{proof}[Proof of Theorem~\ref{coarse graining thm}]
The case $(i)\Rightarrow(ii)$ was proved in \cite[Lemma 3]{rubboli2026strong}. 

We now prove $(ii) \Rightarrow (i)$. Assume that $H^\lambda_{\alpha,z}$ satisfies universal coarse-graining, and suppose for contradiction that $z \neq \alpha$. Let
\begin{equation}
S = 
\begin{pmatrix} 
0 & 1 \\ 
1 & 0 
\end{pmatrix}, \qquad K_s = I + sS, \qquad \lvert s\rvert < 1.
\end{equation}
Since the eigenvalues of $K_s$ are $1 + s$ and $1 - s$, we have $K_s > 0$. Define the faithful state $\tau_s = \frac{K_s^{\frac{z}{1-\alpha}}}{\Tr K_s^{\frac{z}{1-\alpha}}}$. There exists a scalar $\nu_s > 0$ such that $\tau_s^{\frac{1-\alpha}{z}} = \nu_s K_s$. Hence, for every $T \geq 0$,
\begin{equation}
\Psi_{\alpha,z}(T\Vert\tau_s) = \nu_s^z \Tr\left(T^{\frac{\alpha}{2z}} K_s T^{\frac{\alpha}{2z}}\right)^z.
\end{equation}

Let
\begin{equation}
P = \lvert 0\rangle\langle 0\rvert, \qquad \lvert v_t\rangle = \sin t \lvert 0\rangle + \cos t \lvert 1\rangle, \qquad R_t = \lvert v_t\rangle\langle v_t\rvert.
\end{equation}
Define
\begin{equation}
\delta(t,s) = \Tr\left((P + R_t)^{\frac{\alpha}{2z}} K_s (P + R_t)^{\frac{\alpha}{2z}}\right)^z - \Tr\left(P^{\frac{\alpha}{2z}} K_s P^{\frac{\alpha}{2z}}\right)^z - \Tr\left(R_t^{\frac{\alpha}{2z}} K_s R_t^{\frac{\alpha}{2z}}\right)^z.
\end{equation}
Then
\begin{equation}
\Delta_{\tau_s}(P,R_t) = \nu_s^z \delta(t,s).
\end{equation}
Since $\nu_s > 0$, Lemma \ref{coarse graining criterion} implies
\begin{equation}
(\alpha - 1)\delta(t,s) \geq 0.
\end{equation}

We next determine the quadratic part of $\delta(t,s)$ at $(0,0)$. For $s = 0$, we have $K_0 = I$, and since $P$ and $R_t$ are rank-one projections,
\begin{equation}
\delta(t,0) = \Tr(P + R_t)^\alpha - 2.
\end{equation}
The eigenvalues of $P + R_t$ are $1 + \sin t$ and $1 - \sin t$, and therefore
\begin{equation}
\delta(t,0) = (1 + \sin t)^\alpha + (1 - \sin t)^\alpha - 2.
\end{equation}
It follows that
\begin{equation}
\frac{\partial\delta}{\partial t}(0,0) = 0, \qquad \frac{\partial^2\delta}{\partial t^2}(0,0) = 2\alpha(\alpha - 1).
\end{equation}
For $t = 0$, we have $P + R_0 = I$. Since the eigenvalues of $K_s = I + sS$ are $1 + s$ and $1 - s$, while $PK_sP = P$ and $R_0K_sR_0 = R_0$, we obtain
\begin{equation}
\delta(0,s) = (1 + s)^z + (1 - s)^z - 2.
\end{equation}
Hence
\begin{equation}
\frac{\partial\delta}{\partial s}(0,0) = 0, \qquad \frac{\partial^2\delta}{\partial s^2}(0,0) = 2z(z - 1).
\end{equation}
It remains to calculate the mixed derivative. Since $R_t K_s R_t = \left(1 + s\sin(2t)\right)R_t$,
\begin{equation}
\Tr\left(R_t^{\frac{\alpha}{2z}} K_s R_t^{\frac{\alpha}{2z}}\right)^z = \left(1 + s\sin(2t)\right)^z.
\end{equation}
Therefore,
\begin{equation}
\left.\frac{\partial^2}{\partial t\partial s}\Tr\left(R_t^{\frac{\alpha}{2z}} K_s R_t^{\frac{\alpha}{2z}}\right)^z\right|_{t = s = 0} = 2z.
\end{equation}
For the remaining term, using $K_s = I + sS$, we have
\begin{equation}
(P + R_t)^{\frac{\alpha}{2z}} K_s (P + R_t)^{\frac{\alpha}{2z}} = (P + R_t)^{\frac{\alpha}{z}} + s(P + R_t)^{\frac{\alpha}{2z}} S (P + R_t)^{\frac{\alpha}{2z}}.
\end{equation}
Then,
\begin{align}
\left.\frac{\partial}{\partial s}\Tr\left((P + R_t)^{\frac{\alpha}{2z}} K_s (P + R_t)^{\frac{\alpha}{2z}}\right)^z\right|_{s = 0} &= z\Tr\left((P + R_t)^{\frac{\alpha(z - 1)}{z}} (P + R_t)^{\frac{\alpha}{2z}} S (P + R_t)^{\frac{\alpha}{2z}}\right) \\
&= z\Tr\left(S(P + R_t)^\alpha\right).
\end{align}
Moreover,
\begin{equation}
P + R_t = 
\begin{pmatrix} 
1 + \sin^2 t & \sin t\cos t \\ 
\sin t\cos t & \cos^2 t 
\end{pmatrix}, \qquad \left.\frac{d}{dt}(P + R_t)\right|_{t = 0} = S.
\end{equation}
Since $P + R_0 = I$, the derivative of the matrix power at $t = 0$ gives
\begin{align}
\left.\frac{d}{dt}\Tr\left(S(P + R_t)^\alpha\right)\right|_{t = 0} &= \alpha\Tr\left(S(P + R_0)^{\alpha - 1}\left.\frac{d}{dt}(P + R_t)\right|_{t = 0}\right) \\
&= \alpha\Tr\left(SI^{\alpha - 1}S\right) \\
&= 2\alpha.
\end{align}
Therefore,
\begin{equation}
\left.\frac{\partial^2}{\partial t\partial s}\Tr\left((P + R_t)^{\frac{\alpha}{2z}} K_s (P + R_t)^{\frac{\alpha}{2z}}\right)^z\right|_{t = s = 0} = 2\alpha z.
\end{equation}
Combining this with the contribution $2z$ from the $R_t$ term gives
\begin{equation}
\frac{\partial^2\delta}{\partial t\partial s}(0,0) = 2\alpha z - 2z = 2z(\alpha - 1).
\end{equation}

The second-order Taylor expansion of $\delta$ at $(0,0)$ is therefore
\begin{equation}
\delta(t,s) = \alpha(\alpha - 1)t^2 + 2z(\alpha - 1)ts + z(z - 1)s^2 + O\left((\lvert t\rvert + \lvert s\rvert)^3\right).
\end{equation}
The corresponding quadratic-form matrix is
\begin{equation}
M_{\alpha,z} = 
\begin{pmatrix} 
\alpha(\alpha - 1) & z(\alpha - 1) \\ 
z(\alpha - 1) & z(z - 1) 
\end{pmatrix}, \qquad \det M_{\alpha,z} = z(\alpha - 1)(z - \alpha).
\end{equation}
If $0 < \alpha < 1$, condition $(\alpha - 1)\delta(t,s) \geq 0$ requires $M_{\alpha,z}$ to be negative semidefinite. If $\alpha > 1$, it requires $M_{\alpha,z}$ to be positive semidefinite. In both case, $\det M_{\alpha,z} \geq 0$, and hence
\begin{equation}
(\alpha - 1)(z - \alpha) \geq 0.
\end{equation}
On the other hand, since $(\alpha,z,\lambda)\in D$, we have $z \geq \alpha$ for $0 < \alpha < 1$ and $z \leq \alpha$ for $\alpha > 1$. It follows that $z = \alpha$, contradicting the assumption $z \neq \alpha$. This proves $(i)$.
\end{proof}

The preceding characterization becomes more rigid when combined with duality. Indeed, coarse-graining for a single entropy identifies the full club-sandwiched surface $z=\alpha$, whereas requiring the same property for its dual also fixes the optimization parameter.

\begin{corollary}\label{dual-coarse-graining}
Let $(\alpha,z,\lambda)$ and $(\hat{\alpha},\hat{z},\hat{\lambda})$ be dual parameter triples in $\mathcal{D}$. Then the following are equivalent:
\begin{enumerate}
\item[(i)] Both $H_{\alpha,z}^{\lambda}$ and its dual $H_{\hat{\alpha},\hat{z}}^{\hat{\lambda}}$ satisfy universal coarse-graining.
\item[(ii)] $z = \alpha$ and $\lambda=1$, or equivalently, $\hat{z} = \hat{\alpha}$ and $\hat{\lambda} = 1$.
\end{enumerate}
In this case, the dual orders satisfy $1/\alpha + 1/\hat{\alpha} = 2$. Consequently, the optimized sandwiched conditional R\'enyi entropies are precisely the members of the three-parameter family for which universal coarse-graining holds simultaneously for the entropy and its dual.
\end{corollary}


\section{Rigidity of the Duality Relations}
In this section, we establish the uniqueness of the duality relation between the conditional entropies $H_{\alpha,z}^\lambda$ and $H_{\hat{\alpha},\hat{z}}^{\hat{\lambda}}$ on tripartite pure states. Our aim is to show that requiring such a duality to hold universally imposes very rigid constraints on the two parameter triples. In fact, it is enough to test the identity on carefully chosen low-dimensional families of pure states for which the conditional entropies can be evaluated explicitly. Thus, the duality property is already strongly constrained at the level of very simple systems, and these elementary test states suffice to force the full parameter relations necessary for duality.

\begin{theorem}\label{duality-converse}
Let $(\alpha,z,\lambda)$, $(\hat{\alpha},\hat{z},\hat{\lambda}) \in\mathcal{D}$. Suppose that
\begin{equation}
H_{\alpha,z}^{\lambda}(A|B)_{\rho} + H_{\hat{\alpha},\hat{z}}^{\hat{\lambda}}(A|C)_{\rho} = 0 \label{duality-relation}
\end{equation}
for every tripartite pure state $\rho_{ABC}$. Then
\begin{equation}
\frac{z}{1-\alpha} + \frac{\hat{z}}{1-\hat{\alpha}} = 0, \qquad \frac{1-z}{1-\alpha} = \hat{\lambda}, \qquad \frac{1-\hat{z}}{1-\hat{\alpha}} = \lambda. \label{parameter-duality-relations}
\end{equation}
\end{theorem}

\begin{proof}
In order to obtain the first two parameter relations, we consider a first class of states consisting of bipartite pure state tensored with a trivial system.

Let $A$ and $B$ be $d$-dimensional systems, and $C$ be a trivial system. Choose a non-uniform probability distribution $(p_i)_{i=1}^d$ with $p_i>0$, and define $\lvert\psi\rangle_{AB} = \sum_{i=1}^d\sqrt{p_i}\lvert ii \rangle_{AB}$. Consider the tripartite pure state
\begin{equation}
\rho_{ABC} = \lvert\psi\rangle\langle\psi\rvert_{AB} \otimes  \lvert0\rangle\langle0\rvert_C.
\end{equation}
Its relevant marginals are
\begin{equation}
\rho_{AB} = \lvert\psi\rangle\langle\psi\rvert_{AB}, \quad \rho_A = \sum_{i=1}^dp_i\lvert i\rangle\langle i\rvert_A, \quad \rho_B = \sum_{i=1}^dp_i\lvert i\rangle\langle i\rvert_B, \quad \rho_{AC} = \rho_A\otimes\lvert0\rangle\langle0\rvert_C.
\end{equation}
Define $\beta = \frac{z+(1-\lambda)(1-\alpha)}{z-\lambda(1-\alpha)}$. Since $\rho_{AB}$ is pure, Proposition~\ref{pure-state formula} gives $H_{\alpha,z}^{\lambda}(A|B)_{\rho} = -H_{\beta}(\rho_B)$. On the other hand, since $\rho_{AC}$ is a product state, Lemma~\ref{product-state formula} gives $H_{\hat{\alpha},\hat{z}}^{\hat{\lambda}}(A|C)_{\rho} = H_{\hat{\alpha}}(\rho_A)$. Since $\rho_A$ and $\rho_B$ have the same eigenvalues, the duality relation \eqref{duality-relation} therefore implies $H_{\beta}(\rho_A) = H_{\hat{\alpha}}(\rho_A)$. Given that $p$ is nonuniform, Lemma~\ref{renyi-parameter-uniqueness lemma} yields $\beta = \hat{\alpha}$, and hence
\begin{equation}\label{1st-parameter-relation}
\frac{z}{1-\alpha} + \frac{1}{1-\hat{\alpha}} = \lambda.
\end{equation}

By exchanging $A$ and $B$, we also obtain
\begin{equation}\label{2nd-parameter-relation}
\frac{\hat{z}}{1-\hat{\alpha}} + \frac{1}{1-\alpha} = \hat{\lambda}.
\end{equation}

The first class of test states yields two relations involving the parameters, but these are not sufficient to determine all three parameter duality relations completely. We therefore introduce a second class of tripartite states whose $AB$-marginal is quantum-classical, with the tripartite
state chosen as a purification of this marginal. This yields the remaining independent relation among the parameters.

Let $A$ and $C$ be $d+1$-dimensional systems, $B$ be a two-dimensional system, and $p\in(0,1)$. Consider the tripartite pure state
\begin{equation}
\lvert\Psi\rangle_{ABC} = \sqrt{p}\,\lvert000\rangle_{ABC} + \sqrt{\frac{1-p}{d}} \sum_{i=1}^{d}\lvert i1i\rangle_{ABC},
\end{equation}
and let $\rho_{ABC} = \lvert\Psi\rangle\langle\Psi\rvert_{ABC}$. Tracing out $C$, the cross terms vanish, and we get
\begin{equation}
\rho_{AB} = p\lvert0\rangle\langle0\rvert_{A} \otimes \lvert0\rangle\langle0\rvert_{B} + \frac{1-p}{d} \sum_{i=1}^{d} \lvert i\rangle\langle i\rvert_{A} \otimes \lvert1\rangle\langle1\rvert_{B},
\end{equation}
which is a quantum-classical state. Its $AC$-marginal is
\begin{equation}
\rho_{AC} = p\lvert00\rangle\langle00\rvert_{AC} + \frac{1-p}{d} \sum_{i,j=1}^{d} \lvert ii\rangle\langle jj\rvert_{AC}.
\end{equation}
Defining $\lvert\phi\rangle_{AC} = \frac{1}{\sqrt d}\sum_{i=1}^{d}\lvert ii\rangle_{AC}$, we can write $\rho_{AC}$ as $\rho_{AC} = p\lvert00\rangle\langle00\rvert_{AC} + (1-p)\lvert\phi\rangle\langle\phi\rvert_{AC}$.

To evaluate $H_{\alpha,z}^{\lambda}(A|B)_{\rho}$, we introduce a
trivial system $X$, whose unique state is denoted by
$\rho_X$. Consider the state
\begin{equation}
\omega_{AXB} = p\lvert0\rangle\langle0\rvert_A \otimes\rho_X \otimes\lvert0\rangle\langle0\rvert_B + (1-p)\frac{P_d}{d} \otimes \rho_X \otimes \lvert1\rangle\langle1\rvert_B.
\end{equation}
Since $X$ is trivial, adding it to the conditioning system does not change the conditional entropy. Thus
\begin{equation}
H_{\alpha,z}^{\lambda}(A|B)_{\rho} = H_{\alpha,z}^{\lambda}(A|XB)_{\omega}.
\end{equation}
Let
\begin{equation}
P_d:=\sum_{i=1}^{d}\lvert i\rangle\langle i\rvert_A, \qquad \gamma = \frac{1-\alpha}{1-\lambda(1-\alpha)}.
\end{equation}
The system $B$ is classical in $\omega_{AXB}$. Applying Lemma~\ref{classical register lemma}, we obtain
\begin{equation}
H_{\alpha,z}^{\lambda}(A|XB)_{\omega} = \frac{1}{\gamma} \log\left( p\exp\left[ \gamma H_{\alpha,z}^{\lambda}(A|X)_{ \lvert0\rangle\langle0\rvert\otimes\rho } \right] + (1-p)\exp\left[ \gamma H_{\alpha,z}^{\lambda}(A|X)_{\frac{P_d}{d}\otimes\rho} \right] \right).
\end{equation}
Both conditional states are product states. Hence, Lemma~\ref{product-state formula} gives
\begin{equation}
H_{\alpha,z}^{\lambda}(A|X)_{\lvert0\rangle\langle0\rvert\otimes\rho} = 0,
\end{equation}
and Proposition~\ref{pure-state formula} gives
\begin{equation}
H_{\alpha,z}^{\lambda}(A|X)_{\frac{P_d}{d}\otimes\rho} = H_{\alpha}\left(\frac{P_d}{d}\right) = \log d,
\end{equation}
we find
\begin{align}
H_{\alpha,z}^{\lambda}(A|B)_{\rho}
&= \frac{1}{\gamma} \log\left( p\exp(\gamma\cdot0) + (1-p)\exp(\gamma\log d) \right)\\
&= \frac{1}{\gamma} \log\left( p+(1-p)d^{\gamma} \right).
\end{align}

Next, we evaluate $H_{\hat{\alpha},\hat{z}}^{\hat{\lambda}}(A|C)_{\rho}$. Introduce a two-dimensional auxiliary system $Y$, initially prepared in the state $\lvert0\rangle\langle0\rvert_Y$, and define
\begin{equation}
\omega_{ACY} = \rho_{AC}\otimes\lvert0\rangle\langle0\rvert_Y.
\end{equation}
By additivity, adjoining this auxiliary system to the conditioning system does not change the conditional entropy. Hence,
\begin{equation}
H_{\hat{\alpha},\hat{z}}^{\hat{\lambda}}(A|C)_{\rho} = H_{\hat{\alpha},\hat{z}}^{\hat{\lambda}}(A|CY)_{\omega}. \label{adjoin-auxiliary-Y}
\end{equation}
To turn $Y$ into a classical register that distinguishes the two orthogonal components of $\rho_{AC}$, define $P_0=\lvert0\rangle\langle0\rvert_C$, $P_1=\sum_{i=1}^{d}\lvert i\rangle\langle i\rvert_C$,
and $X_Y = \lvert0\rangle\langle1\rvert_Y + \lvert1\rangle\langle0\rvert_Y.$
Consider the operator
\begin{equation}
U_{CY} = P_0\otimes I_Y + P_1\otimes X_Y.
\end{equation}
which can be easily verified to be unitary. Furthermore,
\begin{align}
\widetilde{\omega}_{ACY} = (I_A\otimes U_{CY}) \omega_{ACY} (I_A\otimes U_{CY}^{\dagger}) = p\lvert00\rangle\langle00\rvert_{AC} \otimes\lvert0\rangle\langle0\rvert_Y + (1-p)\lvert\phi\rangle\langle\phi\rvert_{AC} \otimes\lvert1\rangle\langle1\rvert_Y.
\end{align}
Unitary invariance on the conditioning system, and combining \eqref{adjoin-auxiliary-Y} gives
\begin{equation}
H_{\hat{\alpha},\hat{z}}^{\hat{\lambda}}(A|C)_{\rho} = H_{\hat{\alpha},\hat{z}}^{\hat{\lambda}}(A|CY)_{\omega} = H_{\hat{\alpha},\hat{z}}^{\hat{\lambda}} (A|CY)_{\widetilde{\omega}}.
\end{equation}
Let $\hat{\gamma} = \frac{1-\hat{\alpha}} {1-\hat{\lambda}(1-\hat{\alpha})}$. Since $Y$ is classical in $\widetilde{\omega}_{ACY}$, applying Lemma~\ref{classical register lemma} yields
\begin{equation}
H_{\hat{\alpha},\hat{z}}^{\hat{\lambda}}(A|CY)_{\widetilde{\omega}} = \frac{1}{\hat{\gamma}} \log\left(p\exp\left(\hat{\gamma} H_{\hat{\alpha},\hat{z}}^{\hat{\lambda}} (A|C)_{\ketbra{0}{0} \otimes \ketbra{0}{0}}\right) + (1-p)\exp\left(\hat{\gamma} H_{\hat{\alpha},\hat{z}}^{\hat{\lambda}}(A|C)_{\ketbra{\phi}{\phi}}\right)\right).
\end{equation}
Thus, Lemma~\ref{product-state formula} gives
\begin{equation}
H_{\hat{\alpha},\hat{z}}^{\hat{\lambda}} (A|C)_{\ketbra{0}{0}} = 0.
\end{equation}
On the other hand, $\ketbra{\phi}{\phi}_{AC}$ is a maximally entangled state on $d$-dimensional subspaces of $A$ and $C$, and so $\rho_C = \frac{1}{d} \sum_{i=1}^{d}\lvert i\rangle\langle i\rvert_C$. By Proposition~\ref{pure-state formula},
\begin{equation}
H_{\hat{\alpha},\hat{z}}^{\hat{\lambda}} (A|C)_{\ketbra{\phi}{\phi}} = -H_{\hat{\beta}}(\rho_C) = -\log d, \quad \text{where } \hat{\beta} =\frac{\hat{z}+(1-\hat{\lambda})(1-\hat{\alpha})}{\hat{z}-\hat{\lambda}(1-\hat{\alpha})}.
\end{equation}
Therefore, we obtain
\begin{align}
H_{\hat{\alpha},\hat{z}}^{\hat{\lambda}}(A|C)_{\rho} 
&= \frac{1}{\hat{\gamma}} \log\left( p\exp(\hat{\gamma}\cdot0) + (1-p)\exp(-\hat{\gamma}\log d) \right)\\
&= \frac{1}{\hat{\gamma}}\log\left(p+(1-p)d^{-\hat{\gamma}}\right).
\end{align}

Applying the assumed duality relation \eqref{duality-relation} to the state $\rho_{ABC}$, we
obtain
\begin{equation}
\frac{1}{\gamma}\log\left(p+(1-p)d^{\gamma}\right) + \frac{1}{\hat{\gamma}}\log\left(p+(1-p)d^{-\hat{\gamma}}\right) = 0,
\end{equation}
or equivalently,
\begin{equation}\label{log-identity}
    \log\left(p+(1-p)d^{\gamma}\right) = -\frac{\gamma}{\hat{\gamma}}\log\left(p+(1-p)d^{-\hat{\gamma}}\right).
\end{equation}
Let $A(p)=p+(1-p)d^{\gamma}$, $B(p)=p+(1-p)d^{-\hat{\gamma}}$ and $c=-\frac{\gamma}{\hat{\gamma}}$. Since $\gamma\neq0$, $\hat{\gamma}\neq0$ and $(\alpha,z,\lambda),(\hat{\alpha},\hat{z},\hat{\lambda}) \in\mathcal{D}$, we have $c\neq0$. Hence \eqref{log-identity} can be written as
\begin{equation}
A(p)=B(p)^c
\end{equation}
Both $A(p)$ and $B(p)$ are polynomials of degree one in $p$. In particular,
\begin{equation}
A''(p)=B''(p)=0.
\end{equation}
Taking the second derivative with respect to $p$ on both sides of $A(p)=B(p)^c$, we obtain
\begin{align}
0 = \frac{\mathrm d^2}{\mathrm dp^2}B(p)^c = c(c-1)B(p)^{c-2}\bigl(B'(p)\bigr)^2. \label{2nd-eqn}
\end{align}
Since $B(p)>0$ and $B'(p)=1-d^{-\hat{\gamma}}\neq0$ for $d\geq2$ and $\hat{\gamma}\neq0$, and since $c\neq0$, \eqref{2nd-eqn} implies
\begin{equation}
-\frac{\gamma}{\hat{\gamma}} = c = 1.
\end{equation}
Thus, the second class of test states yields the remaining relation
\begin{equation}
\frac{1}{1-\alpha} + \frac{1}{1-\hat{\alpha}} = \lambda+\hat{\lambda}. \label{3rd-parameter-relation}
\end{equation}
Combining \eqref{1st-parameter-relation}, \eqref{2nd-parameter-relation}, and \eqref{3rd-parameter-relation}, we derive \eqref{parameter-duality-relations}, completing the proof.
\end{proof}


\section{Conclusion}

We have established the sharp data-processing region for the three-parameter
conditional R\'enyi entropies $H_{\alpha,z}^{\lambda}$ and shown that, within
this region, the known duality transformation is uniquely determined by its
validity on all tripartite pure states. These results demonstrate that the
corresponding parameter restrictions are intrinsic, rather than artifacts
of the existing proofs. Requiring, in addition, universal monotonicity under
deterministic coarse-graining of the classical conditioned system singles
out precisely the club-sandwiched subfamily $z=\alpha$, providing a
structural characterization of this family through its compatibility
with classical information processing.

\paragraph*{Acknowledgements:}
Y.B. gratefully acknowledges the Centre for Quantum Technologies, National University of Singapore, for hosting his research visit. M.T. acknowledges support from the National Research Foundation Investigatorship Award (NRF-NRFI10-2024-0006). The research is also supported by the National Research Foundation, Singapore, through the National Quantum Office, hosted by A*STAR, under its Centre for Quantum Technologies Funding Initiative (S24Q2d0009).

\bibliographystyle{ultimate}
\bibliography{library}

\end{document}